\documentclass[10pt,a4paper,nofootinbib, notitlepage,prx,aps,twocolumn]{revtex4-2}

\usepackage{graphicx}
\usepackage{dcolumn}

\usepackage{amsmath,amssymb,amsthm,amscd,amsfonts,appendix}
\usepackage{bm, bbm, physics, mathtools, relsize}
\usepackage{xcolor}

\usepackage[most]{tcolorbox}
\usepackage{enumitem}
\usepackage[normalem]{ulem}
\usepackage{algorithm}
\usepackage{algpseudocode}

\usepackage[utf8]{inputenc}

\usepackage{booktabs}
\usepackage{physics}
\usepackage{csquotes}
\usepackage{comment}

\usepackage{ragged2e} 

\makeatletter
\long\def\@makecaption#1#2{%
  \vskip\abovecaptionskip
  \sbox\@tempboxa{\small\textbf{#1.} #2}%
  \ifdim \wd\@tempboxa >\hsize
    {\small\textbf{#1.} \justifying #2\par}%
  \else
    {\hbox to\hsize{\small\textbf{#1.} #2\hfill}}%
  \fi
  \vskip\belowcaptionskip}
\makeatother

\definecolor{babyblue}{rgb}{0.54, 0.81, 0.94}

\usepackage[hidelinks]{hyperref}

\def\H{ {\mathcal H} }

\def\M{ {\mathcal M} }

\def\U{ {\mathcal U} }

\def\N{ {\mathcal N} }

\def\I{ {\mathcal I} }
\def\X{ {\mathcal X} }

\newcommand{\ms}[1]{\textcolor{orange}{#1}}

\def\id{ {\mathbbm 1} }

\newtheorem{proposition}{Proposition}
\newtheorem{theorem}{Theorem}
\newtheorem{lemma}{Lemma}

\newtheorem{alg}{Algorithm}
\newtheorem*{qsearch}{Coh-Search$\bm{(k)}$}

\renewcommand{\eqref}[1]{Eq.~(\ref{#1})}
\newcommand{\figref}[1]{Fig.~\ref{#1}}

\newcommand{\secref}[1]{Section~\ref{#1}}
\newcommand{\appref}[1]{Appendix~\ref{#1}}

\newcommand{\thmref}[1]{Theorem~\ref{#1}}

\begin{document}

\title{Amplitude-amplified coherence detection and estimation}

\date{\today}

\author{Rhea Alexander}
\email{ralexander@ugr.es}
\affiliation{Quantum Thermodynamics and Computation Group. Departamento de Electromagnetismo y Física de la Materia, Universidad de Granada, 18071 Granada, Spain.}
\author{Michalis Skotiniotis}
\affiliation{Quantum Thermodynamics and Computation Group. Departamento de Electromagnetismo y Física de la Materia, Universidad de Granada, 18071 Granada, Spain}
\affiliation{Instituto Carlos I de Física Teórica y Computacional, Universidad de Granada, 18071 Granada, Spain.}
\author{Daniel Manzano}
\affiliation{Quantum Thermodynamics and Computation Group. Departamento de Electromagnetismo y Física de la Materia, Universidad de Granada, 18071 Granada, Spain}
\affiliation{Instituto Carlos I de Física Teórica y Computacional, Universidad de Granada, 18071 Granada, Spain.}

\begin{abstract} 

The detection and characterization of quantum coherence is of fundamental importance both in the foundations of quantum theory as well as for the rapidly developing field of quantum technologies, where coherence has been linked to quantum advantage. Typical approaches for detecting coherence employ {\it coherence witnesses}---observable quantities whose expectation value can be used to certify the presence of coherence. By design, coherence witnesses are only able to detect coherence for some, but not all, possible states of a quantum system.
In this work we construct protocols capable of detecting the presence of coherence in an {\it unknown} pure quantum state $\ket{\psi}$.  Having access to $m$ copies of an unknown pure state $\ket{\psi}$ we show that the sample complexity of any experimental procedure for detecting coherence with constant probability of success $\ge 2/3$ is $\Theta(c(\ket{\psi})^{-1})$, where $c(\ket{\psi})$ is the geometric measure of coherence of $\ket{\psi}$. However, assuming access to the unitary  $U_\psi$ which prepares the unknown state $\ket{\psi}$, and its inverse $U_\psi^\dagger$, we devise a coherence detecting protocol that employs amplitude-amplification {\it a la} Grover, and uses a quadratically smaller number $O(c(\ket{\psi})^{-1/2})$ of samples. Furthermore, by augmenting amplitude amplification with phase estimation we obtain an experimental estimation of upper bounds on the geometric measure of coherence within additive error $\varepsilon$ with a sample complexity that scales as $O(1/\varepsilon)$ as compared to the $O(1/\varepsilon^2)$ sample complexity of Monte Carlo estimation methods.  The average number of samples needed in our amplitude estimation protocol provides a new operational interpretation for the geometric measure of coherence.  Finally, we also derive bounds on the amount of noise our protocols are able to tolerate.
\end{abstract}

\maketitle
\section{Introduction}
Quantum coherence, the ability to prepare quantum states in superpositions over distinct orthogonal pure states such as $\alpha \ket{\text{cat alive}} + (1-\alpha) \ket{\text{cat dead}}$, is a hallmark feature of quantum theory. Coherence serves as a key resource underpinning a range of quantum technologies, such as metrology~\cite{Giovannetti2006Metrology},  cryptography~\cite{Gisin2002Cryptography}, thermodynamics~\cite{Lostaglio2015Thermodynamic,scully:science03,tejero:preprint24}, energy harvesting \cite{dorfman:pnas13,svidzinsky:pra11,manzano:po13}; as well as in the development of quantum algorithms capable of outperforming their classical counterparts~\cite{Hillery2016CoherenceDJ,Liu2019CoherenceAlgorithms,Ahnefeld2022CoherenceShor,Ahnefeld2025CoherencePhaseEst}. 

Generating, detecting, and quantifying coherence has several important applications spanning the fundamental~\cite{Ghoshal2020WitnessGravity,Wagner2024coherence} to the applied~\cite{Diaz2018,Diaz2020,Farré2025WitnessKeyGrowing}. For example, as the first near-term quantum computers come into fruition, benchmarking and verifying their correct functioning, including their ability to maintain superpositions over distinct computational basis states over time, will be ever more important \cite{proctor:nrp25}. \ms{} Moreover, recent schemes for witnessing quantum gravity~\cite{Ghoshal2020WitnessGravity,Matsumura2022CoherenceGravity,Marletto2017Gravity,Lami2024Gravity} based on witnessing coherence or other non-classical resources is both a non-trivial and highly costly task. Reducing the complexity and cost of this task is paramount in bringing quantum technologies closer to fruition.

A brute-force approach to detecting and quantifying quantum coherence is to perform quantum state tomography~\cite{Cramer2010Tomography} in conjunction with classical post-processing of the reconstructed density matrix $\rho$ (e.g.~evaluating some measure of coherence~\cite{Baumgratz2014QuantifyingCoherence,Huang2023ExperimentalRelEnt,Zhang2024Quantifying}). However, full tomographic reconstruction becomes impractical even for modest numbers of qubits, as the sample complexity of quantum state tomography grows quadratically with the Hilbert space dimension. An alternative experimental method for detecting coherence without the need for full tomography proceeds by building a \textit{coherence witness}~\cite{Napoli2016Robustness,Ma2021CoherenceWitness}, inspired by the earlier introduction of entanglement witnesses~\cite{Terhal2000Bell,Lewenstein2000EntanglementWitness}. The basic idea is to construct an observable $W$ such that $\tr(W \sigma) > 0$ for all incoherent states $\sigma$. It follows that, if the experimental reconstruction of the expectation value $\tr(W\rho)$ yields a negative value, then we have detected coherence in $\rho$. This approach works well when the state in question is known. However, any given witness will only detect coherence if the state $\rho$ falls on the correct side of the hyperplane drawn out by $W$, and so for unknown states the construction of sample-optimal witnesses will in general not be possible.

In this work, we address the problem of coherence detection and estimation in unknown pure states, and construct two sample-optimal protocols for performing the former task given slightly different operational assumptions: namely, whether we assume black-box access to copies of the state $\ket{\psi}$ or the unitary $U_\psi$ that  prepares the state $\ket{\psi}$ from some fiducial state $\ket{0}$. We find that having access to the dynamics that prepares a quantum system, as opposed to only having copies of the quantum system itself, allows for the detection and quantification of coherence with quadratically reduced sample complexity. By sample complexity, we refer to the minimal number of copies, $m$, of the state  $\ket{\psi}$ (or unitary $U_\psi$) required on average to verify that $\ket{\psi} $ is coherent. 

It is somewhat intuitive that the sample complexity associated with determining whether a given state is coherent or not will depend on ``how close" the state in question is to the set of incoherent states. We formalize this intuition by establishing a connection between the complexity of detecting coherence in $\ket{\psi}$ and the \textit{geometric measure of coherence}~\cite{Streltsov2015Geometric,Streltsov2017CoherenceReview} $c(\ket{\psi})$, defined as 
\begin{align} \label{eq:geometric_measure_coherence}
  c(\ket{\psi}) \coloneqq 1 -  \max_{\ket{i} \in \I} \abs{\braket{i}{\psi}}^2, 
\end{align}
where the maximization is performed over the set of all incoherent pure states $\I$ of the system. In particular, in the case where we have access to copies of the state $\ket{\psi}$, we show in \thmref{thm:complexity_state} that a number $m = \Theta(c(\ket{\psi})^{-1} \log \delta^{-1})$ of  samples are required to detect coherence with an overall probability of error smaller than $\delta$.  On the other hand, if we instead assume black-box access to the unitary $U_\psi$ and its inverse, we show in  \thmref{thm:coherence_complexity_amp_estimation} that the number of samples required reduces to  $m=\Theta(c(\ket{\psi})^{-1/2} \log \delta^{-1})$. This quadratic reduction in sample complexity can be achieved by exploiting amplitude amplification, a key algorithmic primitive introduced in the seminal works~\cite{Grover1996,Brassard2002OGAmpAmp}. Similarly, by making use of a related algorithm known as amplitude estimation~\cite{Brassard2002OGAmpAmp,Aaronson2020ApproximateCounting,Suzuki2020AmpEst,Grinko2021IterativeAE} we are able to obtain an experimental estimation of upper bounds on the geometric coherence up to additive error $\varepsilon$ with a boosted scaling of $O(1/\varepsilon)$ with respect to the $O(1/\varepsilon^2)$ scaling expected from Monte Carlo estimation.

The remainder of the paper is structured as follows. In \secref{sec:hyp_testing}, we connect the problem of coherence detection to hypothesis testing. In Sections~\ref{sec:state_access} and \ref{sec:amp_amp_coherence_detection}, we derive our results on the query complexity of coherence detection when we have access to states and unitaries respectively, in each case providing explicit protocols for coherence detection. In \secref{sec:coherence_estimation}, we provide a simple protocol yielding upper bounds to the geometeric measure of coherence, which are tight in a typical experimental run. In \secref{sec:robustenss_to_noise}, we consider the robustness of our protocol for amplitude amplified coherence detection to the effects of noise. We summarize and conclude in \secref{sec:conclusions}.

\begin{figure}[t]
	\centering	\includegraphics[width=0.8\linewidth]{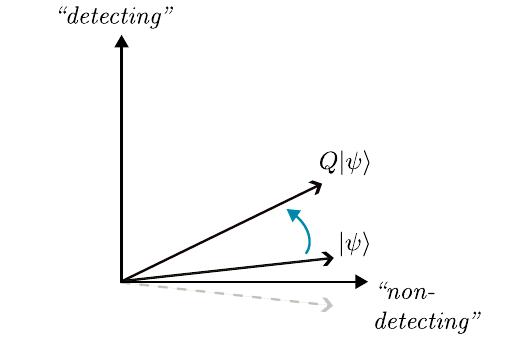}
\caption{\textbf{(Amplitude-amplified coherence detection).} 
	Schematic representing the intuition behind amplitude-amplified coherence detection.	We rotate an initial coherent state $\ket{\psi}$ in a 2-dimensional subspace towards an especially chosen subspace which allows us to detect coherence.
        \label{fig:amp_amp}}
\end{figure} 

\section{Coherence detection as a hypothesis testing problem}
\label{sec:hyp_testing}

Consider a $d$-dimensional quantum system with associated Hilbert space $\H$. Let $\I \coloneqq \{\ket{0},\ket{1},\dots,\ket{d-1}\}$ be the set of all incoherent pure state vectors, which we shall assume form a complete orthonormal basis for $\H$. 
We can decompose any pure state $\ket{\psi}$ in $\H$ in the incoherent basis $\I$ as
\begin{align}
    \ket{\psi} \coloneqq \sum_k e^{i \phi_k} \sqrt{p_k} \ket{k},
\end{align}
for some probability distribution $\{p_k\}$ and  collection of phases $\{\phi_k \in [0,2\pi)$\}. The pure state $\ket{\psi}$ is incoherent whenever there is only one non-zero term in the above sum, and otherwise it is coherent. One way of quantifying the amount of coherence in such a state is via the construction of measures of coherence (e.g. see the review~\cite{Streltsov2017CoherenceReview}), functions from the space of quantum states to the reals that cannot increase under the action of channels which do not generate coherence. This demand corresponds to the physical statement that any coherence non-generating channel can only degrade (or keep constant) the amount of coherence present in a given state. One such coherence measure in pure states is the geometric measure of coherence defined in \eqref{eq:geometric_measure_coherence}.

The central question we seek to address in this work is:
\begin{tcolorbox}[breakable,colback=babyblue!10!white,colframe=babyblue!80!black]
\begin{center}
\textbf{Question:} 
\textit{What is the sample complexity of detecting coherence in an unknown pure state $\ket{\psi}$?
}
\end{center}
\end{tcolorbox}

To address this question we proceed by rephrasing the problem of coherence detection as the following binary hypothesis testing scenario:
\begin{itemize}
\item[(H0)] \textit{(Incoherent hypothesis).} $\ket{\psi}$ is incoherent, meaning that, up to some global phase, it is equal to the incoherent state $\ket{k}$ for some $k \in \{ 0,1,\dots, d-1\}$.
\item[(H1)] \textit{(Coherent hypothesis).} There exists some fixed $\varepsilon >0$ such that $\ket{\psi}$ satisfies:
\begin{align} \label{eq:H1_trace_distance}
\min_{\ket{k} \in \I} D(\psi , \ketbra{k})  \ge \varepsilon,
\end{align}
 where $D(\rho, \tau) \coloneqq \frac{1}{2}\norm{\rho - \tau}_1$ is the trace distance, defined in terms of the trace norm $\norm{X}_1 \coloneqq \tr \sqrt{X^\dagger X}$.
\end{itemize}

When sufficient information regarding the underlying state $\ket{\psi}$ is known, in particular, if we know which incoherent state $\ket{k_{\max}}$ achieves the minimum in \eqref{eq:H1_trace_distance} (it's ``closest" incoherent state), then the problem of coherence detection is equivalent to verifying that the state $\ket{\psi}$ is not of the form $e^{i \phi}\ket{k_{\max}}$, for any globally irrelevant phase $\phi$. This is a task known in the property testing literature as \textit{equality testing}~\cite{Kada2008IdentityTesting,Montanaro2018Property}. On the other hand, for completely unknown quantum states, the problem at hand is a harder one, since in this case we do not know the precise value of $\varepsilon$ or the label 
$k_{\mathrm{max}}$. Despite this, as we shall see in the next section, this perspective produces lower bounds on the sample complexity of coherence detection, even when the state in question is unknown.

\section{Complexity of pure state coherence detection}
\label{sec:state_access}

Here we consider the sample complexity of coherence detection  given access to $m$ copies of the pure state $\ket{\psi}$. The following theorem, which involves standard binary quantum hypothesis testing (see~\cite{Helstrom1976quantum} as well as the excellent review on quantum property testing~\cite{Montanaro2018Property}), establishes the sample complexity of coherence detection.

\begin{tcolorbox}[breakable,colback=babyblue!10!white,colframe=babyblue!80!black]
\begin{theorem} \label{thm:complexity_state} Detecting coherence in the unknown pure state $\ket{\psi} \in \H$ up to fixed error probability $\delta < \frac{1}{2}$ requires a number of copies
\begin{align}
    m   = \Theta\left( \frac{1}{c(\ket{\psi})} \log \delta^{-1} \right),
\end{align}
on average. This includes the case when we are allowed to perform global measurements on the state $\ket{\psi}^{\otimes m}$.
\end{theorem}
\end{tcolorbox}

\begin{proof}
We begin by first establishing a lower bound on the sample complexity of detecting coherence in the state $\ket{\psi}$. 
Note that the minimal probability of error $P_{\mathrm{err}}$ associated with any protocol for detecting coherence when the state $\ket{\psi}$ is unknown is greater than or equal to the minimal probability of error when we have knowledge regarding the state $\ket{\psi}$. This is because any protocol from the former set can be simulated by the latter set by simply disregarding the information. Moreover, given knowledge of the closest incoherent state $\ket{k_{\max}}$ to $\ket{\psi}$, the problem of coherence detection reduces to equality testing. Any (potentially global) experiment which consumes $m$ copies of $\ket{\psi}$ to test equality between $\ket{\psi}$ and $\ket{k_{\max}}$ is equivalent to a protocol for discriminating $\ket{\psi}^{\otimes m}$ from $\ket{k_{\max}}^{\otimes m}$. In such a scenario, the minimal probability of error $P^*_{\mathrm{err}}$ in distinguishing between $\ket{\psi}^{\otimes m}$ and $\ket{k_{\mathrm{max}}}^{\otimes m}$ (with equal prior probability) is given by~\cite{Helstrom1976quantum,Montanaro2018Property} 
    \begin{align}
        P^*_{\mathrm{err}} &=\frac{1}{2} \left( 1 - D\left(\psi^{\otimes m}, \ketbra{k_{\mathrm{max}}}^{\otimes m}\right)\right) \notag \\
        &= \frac{1 }{2} \left(1 - \sqrt{ 1- \abs{\braket{k_{\mathrm{max}}}{\psi}}^{ 2m} } \right) \notag \\
        &= \frac{1 }{2} \left(1 - \sqrt{ 1- p_{k_{\max}}^m } \right)\, ,
    \end{align}
This error probability is achieved by the Helstrom
measurement~\cite{Helstrom1976quantum}: a measurement whose operators consist of 
projectors onto the positive and negative eigenspaces of $\psi^{\otimes m}-
\ketbra{k_{\mathrm{max}}}^{\otimes m}$. Imposing the requirement that 
$P^*_\mathrm{err}\leq P_{\mathrm{err}} \leq \delta$ gives
    \begin{align}
        \delta \ge \frac{1}{2} \left(1 - \sqrt{ 1- p_{k_{\max}}^m } \right),
    \end{align}
which in turn implies 
    \begin{align}
        p_{k_{\max}}^m \le 4 \delta (1 -\delta).
    \end{align}
As $p_{k_{\max}}^m\equiv \left(1 - c(\ket{\psi}) \right)^m$, by taking natural logarithms on both sides and rearranging we obtain
    \begin{align}
        m \ge \frac{\ln [4\delta(1-\delta))]^{-1}}{\ln [1-c(\ket{\psi})]^{-1}}  \ge \frac{\ln [4\delta(1-
        \delta))]^{-1}}{c(\ket{\psi})},
    \end{align}
where the second inequality follows from (a rearrangement of) the identity $1-x \le e^{-x}$. For fixed $\delta$ 
this can be written as
\begin{align}
    m \ge \frac{1}{c(\ket{\psi})} \ln \frac{1}{C \delta},
\end{align}
where $C = 4(1-\delta)$ is a small constant $2< C \le 4$. Therefore, we can conclude that 
$m  = \Omega \left( c(\ket{\psi})^{-1} \log(1/\delta) \right)$, establishing a lower bound on the sample 
complexity.

To establish a corresponding upper bound we consider the following protocol 
consisting of independently measuring each copy of $\ket{\psi}$ in the incoherent 
basis $\I$. If the sequence of measurement 
outcomes contains more than a single value $k$, then we know with certainty that 
the state is coherent. On the other hand, if after $m$ measurements all outcomes 
are identical, we conclude that the state is incoherent.  
The probability of error for this protocol is given by 
\begin{align}
    P_{\mathrm{err}} \leq p_{k_{\max}}^m = \left[1-c(\ket{\psi}) \right]^m \le e^{-m c(\ket{\psi})},
\end{align}
where the second inequality is due to the bound $1-x \le e^{-x}$. To achieve a 
failure probability at most $\delta$, it is sufficient to take $m = \left\lceil c(\ket{\psi})^{-1} \ln (1/\delta) \right\rceil$ establishing the upper bound $m  = O\left(  c(\ket{\psi})^{-1} \log (1/\delta) \right)$.  As both upper and lower bounds grow identically it follows that the sample complexity is exactly $\Theta(c(\ket{\psi})^{-1} \log (1/\delta))$ completing the proof. 
\end{proof}

In the next section we will show how to reduce the sample complexity in coherence detection by making use of amplitude amplification.

\section{Amplitude-amplified coherence detection}
\label{sec:amp_amp_coherence_detection}

Here we construct a protocol for detecting the presence of coherence in a state\footnote{We note that this choice is without loss of generality, since the fixed fiducial state $\ket{0}$ can be taken to be anything we choose simply by modifying the unitary $U_\psi$ appropriately, and thus need not be a member of the incoherent basis.} $\ket{\psi}=U_{\psi}\ket{0}$ assuming that we have access to both the unitary dynamics $U_\psi$ and its inverse. More precisely, we consider the following setting: given some initial state $\ket{\phi} \coloneqq \ket{0}^{\otimes \mathrm{poly}(n)}$, we can apply our black-box unitary $U \in \{ U_\psi , U_\psi^\dagger\}$ interleaved with unitary processing $V_i$:
\begin{align}
\ket{\phi_U} \coloneqq  V_m U V_{m-1} \dots  U V_1 U V_0 \ket{\phi},
\end{align}
followed by a final (potentially global) measurement $\M = \{M_{\mathrm{yes}},M_{\mathrm{no}} \}$ on the state $\phi_U$ which decides whether $\psi$ is coherent.  
We construct a particular protocol of this form based on amplitude amplification, which we describe in \secref{sec:amp_amp_est}, before providing our coherence detection algorithm in \secref{sec:amplitude_amplified_coherence_detection}, and an analysis of its corresponding sample complexity in \secref{subsection:sample_complexity_unitary}.

\subsection{Subroutine for amplitude amplified equality-testing}
\label{sec:amp_amp_est}

Before outlining our full protocol for witnessing coherence in an \textit{unknown} pure state $\ket{\psi}$, let us first review a protocol for distinguishing the pure state $\ket{\psi}$ from any fixed incoherent basis state $\ket{k}$ under the promise that
\begin{align}
  p_k =  \abs{\braket{k}{\psi}}^2 > 0.
\end{align}
Performing the POVM $\M_k \coloneqq \{ \Pi_k , \Pi_{\overline{k}} \}$, where $\Pi_k \coloneqq \ketbra{k}$ and $\Pi_{\overline{k}} \coloneqq \id - \Pi_k$, allows for the detection of coherence whenever the outcome $\overline{k}$ is obtained. The success probability associated with $\M_k$ as a coherence detection device is $1-p_k$ and the corresponding sample complexity is $(1-p_k)^{-1}$. 

To reduce the sample complexity of the above protocol we would like to increase the 
success probability associated with $\M_k$. This can be 
achieved via amplitude amplification~\cite{Grover1996,Brassard2002OGAmpAmp}.  The basic idea, illustrated in \figref{fig:amp_amp}, is to rotate the initial state $\ket{\psi}$ away from the subspace spanned by $\ket{k}$. Specifically, let us partition the state space into the direct sum $\H_S=  \H_k \oplus \H_{\overline{k}}$ of two mutually orthogonal subspaces: the ``informative'' or ``detecting'' subspace $\H_{\overline{k}}:=\mathrm{Im}(\Pi_{\overline{k}})$ and the ``non-detecting'' subspace $\H_k:=\mathrm{Im}(\Pi_k)$. Without loss of generality, and up to some globally irrelevant phase, we can always write
\begin{align} \label{subeq:phi_a_AA}
    \ket{\psi} \coloneqq \sqrt{p_k} \ket{k} + \sqrt{1-p_k }\ket{k_\perp},
\end{align}
for some state vector $\ket{k_\perp}\in\H_{\overline{k}}$.

Geometrically, we can think of each iteration of an amplitude amplification protocol as building a unitary rotation towards the detecting subspace from a sequence of two reflections. More precisely, one proceeds by constructing the so-called \textit{Grover operator} $Q_k \coloneqq  V_\psi V_k$ from the two unitaries:
\begin{align} \label{eq:grover_1}
 V_k &\coloneqq \id - 2 \ketbra{k}, \\
V_\psi &\coloneqq \id - 2 \ketbra{\psi} = U_\psi( \id - 2 \ketbra{0}) U_\psi^\dagger.
\label{eq:grover_2}
\end{align}
Writing $p_k = \cos^2 (\theta_k)$, $m$ applications of the Grover operator $Q_k$ on the state vector $\ket{\psi}$ results in a state vector of the form
\begin{align}
    Q^m_k &\ket{\psi} = \notag \\ &\cos((2m+1)\theta_k)\ket{k} + \sin((2m+1)\theta_k) \ket{k_\perp}, \label{eq:grover_state_m}
\end{align}
which consumes $m+1$ calls to the oracle $U_\psi$ and $m$ calls to its inverse $U_\psi^\dagger$. By inspection of \eqref{eq:grover_state_m},  the probability to detect coherence in the state $\ket{\psi}$ is 
\begin{align} \label{eq:p_detect}
 P_{\mathrm{suc}} =  \sin^2((2m+1)\theta_k),
\end{align}
which can be made close to $1$ by setting $m = \lfloor \pi/4\theta_k - 1/2 \rfloor$.
Note, however, that we do not assume to know the parameter $\theta_k$ (or equivalently $p_k$). Therefore the precise number of Grover rotations required to optimize this probability is also unknown. Thankfully, the following modified protocol from Ref.~\cite{Brassard2002OGAmpAmp} works regardless of whether the parameter $\theta_k$ is known, and manages to maintain the same asymptotic scaling in oracle calls as standard amplitude amplification.

\begin{qsearch}[\cite{Brassard2002OGAmpAmp}]
\textbf{Inputs:} Label $k \in \{0,1,\dots d-1\}$, black-box access to $U_\psi$ and $U_\psi^\dagger$, with the promise that $ \abs{\braket{\psi}{k}}^2 > 0$.

\textbf{Returns:} ``coherent'' label only if the state $\ket{\psi}$ is coherent.
\begin{itemize}
    \item[1:] Set $\ell \leftarrow 0$ and $c \leftarrow 3/2$.
    \item[2:] \textbf{Repeat:}
    \begin{itemize}
        \item Update $\ell \leftarrow \ell +1$. Set $M= \lceil c^\ell \rceil$.
        \item Prepare $U_\psi \ket{0}$ and measure in the incoherent basis. If the outcome is not labeled by $k$ output \textit{``coherent"} and \textbf{stop}.
        \item Pick an integer $j$ uniformly at random from $\{1,2,\dots, M\}$.
        \item Prepare $Q^j_k U_\psi \ket{0}$ and measure in the incoherent basis. If the outcome is not labeled by $k$ output \textit{``coherent"} and \textbf{stop}.
    \end{itemize}
\end{itemize}
\end{qsearch} 

This algorithm performs an iteratively increasing number of Grover rotations according to an exponential schedule to ensure that one ``lands'' closer to the detecting subspace with high probability. More formally, from \eqref{eq:p_detect} the success probability of detecting coherence after $j$ applications of the Grover operator $Q_k$ is given by $\sin^2 \left( (2j +1 ) \theta_k\right)$, and therefore the average success probability when $0 \le j < M$ is chosen uniformly at random is
\begin{align}
    P_{\mathrm{suc}} &=  \frac{1}{M} \sum_{j=0}^{M-1}  \sin^2 \left( (2j +1 ) \theta_k\right) \notag \\ &= \frac{1}{2} - \frac{\sin(4 M \theta_k)}{4 M \sin (2 \theta_k)} \ge\frac{1}{2} - \frac{1}{4 M \sin (2 \theta_k)} ,
\end{align}
where the second equality is shown in Lemma~1 of~\cite{Boyer1998TightBounds}. 

\begin{figure}[t]
	\centering	\includegraphics[width=0.999\linewidth]{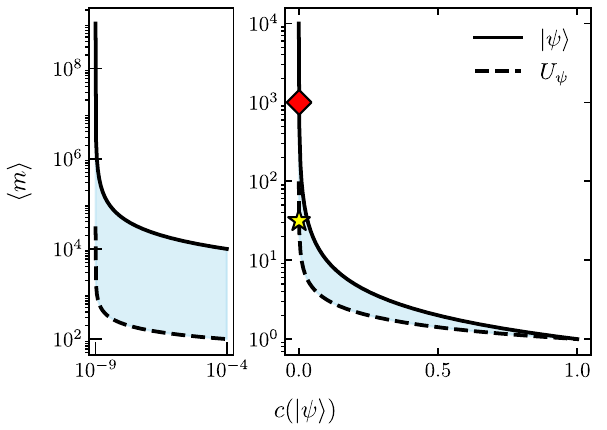}
\caption{ \textbf{(Complexity of coherence detection).} We plot the optimal sample complexity scaling of detecting coherence in the pure state $\ket{\psi}$ as a function of coherence of the state in question, as measured by the geometric measure of coherence $c(\ket{\psi})$. The straight line depicts the lower bound from \thmref{thm:complexity_state} supposing access to $m$ copies the state $\ket{\psi}$ on average. In contrast, as shown in \thmref{thm:coherence_complexity_amp_estimation}, given access the unitary, $U_\psi$, which prepares $\ket{\psi}$, amplitude-amplified coherence detection provides a quadratic improvement (dashed line). The star and diamond guide the eye to the difference between the two curves at $c(\ket{\psi}) = 10^{-3}$.
		\label{fig:complexity_coherence}}
\end{figure}

It follows that for any $M\ge 1/\sin(2\theta_k)$ the probability of success is bounded as $P_{\mathrm{suc}} \ge 1/4$.  Once this critical parameter regime has been reached in the exponential schedule, a single iteration of step~2 of the above algorithm yields a coherence detection result with a constant probability at least $1/4$, making detection likely after just a few iterations of this step. Therefore, the correctness of this algorithm for detecting coherence is immediate, all we need now is to evaluate its expected sample complexity. Importantly, the overhead of \textbf{Coh-Search$\bm{(k)}$} scales with the same asymptotic overhead as the optimal amplitude amplification protocol when $\theta_k$ is known.

\begin{proposition}[Theorem~3 of~\cite{Brassard2002OGAmpAmp}] \label{prop:coh_search_scaling}
   Suppose that $\ket{\psi} \in \H$ satisfies $\abs{\braket{\psi}{k}}^2>0$ for a known label $k$.  Then the algorithm \textup{\textbf{Coh-Search$\bm{(k)}$}} detects coherence in the state $\ket{\psi}$ using an expected number $O (1/\sqrt{1-p_k})$ of applications of $U_\psi$ and $U_\psi^\dagger$ whenever $p_k <1$ with probability of success $\ge 1/4$, and otherwise runs forever.
\end{proposition}

The proof of this proposition is equivalent to that of Theorem~3 of \cite{Brassard2002OGAmpAmp}. The intuition is that the scaling will be dominated by the number of Grover applications in the final stage of the algorithm $m = O(M_{\mathrm{final}})$, which is typically reached in the critical regime where we have $M_{\mathrm{final}} = O\left(  1/\sin(2 \theta_k)  \right) = O\left( 1/ \sqrt{1-p_k} \right)$.

\subsection{Algorithm for amplitude-amplified coherence detecting}
\label{sec:amplitude_amplified_coherence_detection} 

Our algorithm for amplitude-amplified coherence detection is as follows.

\begin{alg}[Amplitude-amplified coherence detection] \label{algorithm:coherence_detection}
\textbf{Inputs:} Black-box access to $U_\psi$ and $U_\psi^\dagger$.

\textbf{Returns:} ``coherent'' label only if the state $\ket{\psi}$ is coherent.

\begin{itemize}
    \item[\textbf{1:}] With a small constant number $C = O(1)$ of calls to the oracle measure $C$ copies of the state $\ket{\psi}$ in the incoherent basis. 
    \begin{itemize}
        \item  \textbf{If} any of the measurement outcomes is distinct: \textbf{return} ``coherent'' and \textbf{stop}.
        \item \textbf{Else}: construct the Grover operator $Q_k \coloneqq  V_\psi V_k$, where $k$ labels the measurement outcome obtained in each of the $C$ measurements, and proceed to step~2. 
    \end{itemize}
  
    \item[\textbf{2:}] Call \textup{\textbf{Coh-Search$\bm{(k)}$}} 
    \begin{itemize}
        \item \textbf{If} \textup{\textbf{Coh-Search$\bm{(k)}$}} terminates: \textbf{return} ``coherent''.
    \end{itemize}
      
\end{itemize}
\end{alg}

Notice that whenever Algorithm~\ref{algorithm:coherence_detection} proceeds to step~2, the outcome labelled by $k$ obtained in the inital $C$ measurements implies that $p_k=\abs{\braket{k}{\psi}}^2>0$. This ensures that the promise required to call \textbf{Coh-Search$\bm{(k)}$} is always guaranteed to be satisfied.

\subsection{Sample complexity scaling}
\label{subsection:sample_complexity_unitary}

 By analysing the average number of calls to the oracle consumed in Algorithm~\ref{algorithm:coherence_detection} and exploiting a recent result~\cite{Weggemans2025lowerboundsunitary} from the field of unitary channel discrimination~\cite{Acin2001UnitaryDistinguishability,Duan2007UnitaryDiscrim}, we arrive at the following theorem.

\begin{tcolorbox}[breakable,colback=babyblue!10!white,colframe=babyblue!80!black]
\begin{theorem} \label{thm:coherence_complexity_amp_estimation} Given black-box access to $U \in \{U_\psi, U_\psi^\dagger\}$, detecting coherence in the unknown state $\ket{\psi} = U_\psi \ket{0} \in \H$, up to fixed error probability $\delta < \frac{1}{2}$, requires 
\begin{align}
    m  = \Theta\left(  \frac{1}{\sqrt{c(\ket{\psi})}}\log \delta^{-1} \right),
\end{align}
 calls to $U$ on average.

\end{theorem}
\end{tcolorbox}

\begin{proof}[Upper bound proof]

In the best case, Algorithm~\ref{algorithm:coherence_detection} terminates with just $m=2$ calls to the oracle, which occurs when $\ket{\psi}$ is indeed coherent and we manage to detect coherence in step~1 of the algorithm. However, with probability $p_k^C$, where $k$ labels the outcome obtained in each of the $C$ measurements, we do not detect coherence in step~1 and proceed to step~2. Therefore, the expected number of calls to the oracle of \textbf{Algorithm~\ref{algorithm:coherence_detection}} is given by
\begin{align}
m  &= O\left( \sum_k \left[ p_k^C (  N_k +C)   + (1-p_k^C) C \right] \right).
\end{align}
Taking $C = 1$ this expression simplifies to the number of average steps with which step~2 terminates, i.e.,
\begin{align} 
 m &= O\left( \sum_k p_k N_k \right) .
\end{align}
By Proposition~\ref{prop:coh_search_scaling}, each call to \textbf{Coh-Search$\bm{(k)}$} terminates with $N_k  = O( 1/\sqrt{1-p_k})$ calls to the oracle with probability of success $\ge 1/4$. Therefore, 
\begin{align} \label{subeq:avg_m_N}
 m =  O \left( \sum_k  \frac{p_k}{\sqrt{1-p_k}}  \right) .
\end{align}
 Now,
\begin{align}
    \sum_k  \frac{p_k}{\sqrt{1-p_k}} 
&= \frac{p_{k_{\max}}}{\sqrt{1-p_{k_{\max}}}} + \sum_{j \neq k_{\max}}  \frac{p_j}{\sqrt{1-p_j}} \notag \\
&= \frac{1- c(\ket{\psi})}{\sqrt{c (\ket{\psi})}} + \sum_{j \neq k_{\max}}  \frac{p_j}{\sqrt{1-p_j}} \label{subeq:p_sums1} ,
\end{align}
Now, for all $j \neq k_{\max}$, we must have $p_j \le \frac{1}{2}$, and therefore
\begin{align}
\sum_{j \neq k_{\max}}  \frac{p_j}{\sqrt{1-p_j}} &\le \sum_{j \neq k_{\max}}  \frac{p_j}{\sqrt{1- \frac{1}{2}}} = \sqrt{2}\sum_{j \neq k_{\max}} p_j \notag \\
&= \sqrt{2} (1 - p_{k_{\max}}) = \sqrt{2} c(\ket{\psi}) .\label{subeq:p_sums2} 
\end{align}
Combining Eqs.~(\ref{subeq:p_sums1}) and (\ref{subeq:p_sums2}) we obtain the bound
\begin{align}
\sum_k \frac{p_k}{\sqrt{1-p_k}} &\le \frac{1}{\sqrt{c(\ket{\psi}) }} + \sqrt{c(\ket{\psi}) } + \sqrt{2} c(\ket{\psi}) \notag \\
&\le \frac{1}{\sqrt{c(\ket{\psi}) }} + 1+ \sqrt{2} . \label{subeq:bound_sum_k}
\end{align}
Combining Eqs.~(\ref{subeq:avg_m_N}) and (\ref{subeq:bound_sum_k}), we can conclude that 
\begin{align}
  m  =    O\left(\frac{1}{\sqrt{c(\ket{\psi}) }}\right),
\end{align}
calls to $U_\psi$ and $U_\psi^\dagger$ allow one to detect coherence  with constant probability $\ge 1/4$. 

Finally, running this algorithm in parallel with $r=O(\log(1/\delta))$ executions, terminating whenever one of the individual runs terminates, boosts the success probability to $\ge 1-\delta$, with the associated sample complexity 
\begin{align}
 m  =    O\left(\frac{1}{\sqrt{c(\ket{\psi}) }} \log \delta^{-1}\right),
\end{align}
as claimed. \end{proof}

The lower bound proof follows much the same idea the corresponding proof of \thmref{thm:complexity_state}, but instead exploiting established results in unitary discrimination~\cite{Weggemans2025lowerboundsunitary}, where the quantum objects one seeks to discriminate are unitary channels as opposed to quantum states. A proof can be found in \appref{appx:lower_bound_proof_unitary}.

Our lower bound assumes that any intermediate processing in between each black-box call to $U_\psi$ or $U_\psi^\dagger$ is unitary. However, recent results have shown that the query complexity of Boolean functions cannot be reduced under general higher order quantum computations~\cite{Abbot2024Complexity_Indefinite}, which for instance, illustrates that the scaling Grover's algorithm cannot be improved even under quantum circuits which involve superpositions or coherent control of the order of queries. Given such results, it may be the case that our lower bound also holds in the more general setting where our adaptive strategy is constructed from a sequence of general quantum channels. Alternative routes to generalizing our lower bound could draw on recent results in quantum channel discrimination~\cite{huang2025querycomplexitiesquantumchannel}. However, we consider such an extension to be outside of the scope of the current work.

A comparison of the sample complexity scaling achievable under our two protocols for coherence detection is shown in \figref{fig:complexity_coherence}.

\section{Coherence estimation}
\label{sec:coherence_estimation}

Aside from detecting the presence of coherence, an important associated task is to be able to quantify the amount of coherence in a given state. Here we show that it is possible to estimate an upper
bound to the amount of coherence present in an unknown pure state by using the well-known protocol of {\it amplitude estimation}~\cite{Brassard2002OGAmpAmp,Aaronson2020ApproximateCounting,Suzuki2020AmpEst,Grinko2021IterativeAE}. Specifically, we seek to estimate the unknown parameter $c_k \coloneqq 1-p_k$ in \eqref{subeq:phi_a_AA}, which satisfies $c_k \ge c(\ket{\psi})$ for all $k \in \{0,\dots, d-1\}$. To do so we can use the standard amplitude estimation protocol in~\cite{Brassard2002OGAmpAmp} to perform phase estimation~\cite{NielsenChuang2010} on the Grover operator $Q_k$ defined by Eqs.~(\ref{eq:grover_1}) and (\ref{eq:grover_2}). As shown in Theorem~12 of Ref.~\cite{Brassard2002OGAmpAmp}, with probability at least $8 / \pi^2$ the error in our estimate $\hat{c}_k$ using this procedure is given by
\begin{align}
    \varepsilon \le \frac{2 \pi \sqrt{c_k (1 - c_k)}}{m} + \frac{\pi^2}{m^2}.
\end{align}
Moreover, this success probability can be readily increased to near unity by repeating the procedure multiple 
times---obtaining the collection of estimates $\{\hat{c}_k^1, \hat{c}_k^2, \dots \hat{c}_k^\ell\}$---and taking the estimate corresponding to the median value.

\begin{alg}[Amplitude-amplified coherence estimation] \label{algorithm:coherence_estimation}
\textbf{Inputs:} Black-box access to $U_\psi$, $U_\psi^\dagger$, and $cU_\psi$. Desired accuracy $\varepsilon > 0$.

\textbf{Returns:} estimate $\hat{c}_k$ of $c_k= 1 - \abs{\braket{k}{\psi}}^2$. 

\begin{itemize}
    \item[\textbf{1:}] With a small constant number $C = O(1)$ of calls to the oracle measure $C$ copies of the state $\ket{\psi}$ in the incoherent basis. 
 
      \item[\textbf{2:}]  Select an outcome label $k$ which was observed most times in the measurement and construct the Grover operator $Q_k = V_\psi V_k$. 

    \item[\textbf{3:}] Perform amplitude estimation on $Q_k$ to produce an estimate $\hat{c}_k$ of $c_k$ such that $\abs{\hat{c}_k - c_k} \le \varepsilon$.

\end{itemize}
\end{alg}

The sample complexity of this algorithm scales as $m=O(1/\varepsilon)$, which provides a quadratic improvement over the expected value from classical Monte-Carlo sampling which requires $O(1/\varepsilon^2)$. We also highlight that on the typical run of this protocol the estimate returned will actually coincide with an estimate of $c(\ket{\psi})$, due to the identity $c(\ket{\psi}) = c_{k_{\mathrm{max}}}$.

\section{Robustness to noise}
\label{sec:robustenss_to_noise}

Amplitude amplification and amplitude estimation are  typically considered within the regime of fault-tolerant quantum computing. This is due to the fact that large depth quantum circuits are usually required, and as such a quadratic speed-up may be quickly washed out by the presence of noise. In this section we study the robustness of our coherence detection and estimation protocols against uncorrectable noise. Specifically, we consider the regime of finite but small probability of an error  occurring after each circuit layer. While much more comprehensive error analyses can be found for amplitude amplification and estimation with noise~\cite{Brown2020ampestNoise,Tanaka2021AestNoise,Herbert2024AmpEstNoise}, here we provide a simple back-of-the-envelope calculation to identify the general parameter regime where we expect our protocol to function.

To model noise which might occur on the system, we imagine after each call to the oracle, labeled by $i \in [m]$, our system is subject to some noise channel 
\begin{align}
  (1-p_{\mathrm{err}}) \mathrm{id} + p_{\mathrm{err}} \N_i,
\end{align}
corresponding to the situation where with probability $(1-p_{\mathrm{err}})$ no noise occurs and with probability 
$p_{\mathrm{err}}$ the system is subjected to the noise channel $\N_i$.
We remain agnostic about the precise action of the noise at each iteration, specified by the channel $\N_i$, and simply concern ourselves with the probability $p_{\mathrm{err}}$ with which noise is likely to occur on each Grover iteration.

After $m$ calls to the oracle, $\psi^{\mathrm{out}}$ is the final state of the noiseless algorithm, the final state of the noisy circuit will be of the form
\begin{align}
    (1-p_{\mathrm{err}})^m \psi^{\mathrm{out}} + \left(1 - (1-p_{\mathrm{err}})^m  \right) \sigma^{\mathrm{junk}},
\end{align}
for some state $\sigma^{\mathrm{junk}}$.
By the binomial theorem, we can compute
\begin{align}
    (1-p_{\mathrm{err}})^m &= \sum_{k=0}^m (-1)^k {\binom{m}{k}} p_{\mathrm{err}}^k \notag \\
&= 1 - m p_{\mathrm{err}} + O(p_{\mathrm{err}}^2).
\end{align}

Therefore, if our original success probability is $P_{\mathrm{suc}}$, incorporating noise we find that so long as the modified success probability $\tilde{P}_{\mathrm{suc}}$ satisfies
\begin{align}
  \tilde{P}_{\mathrm{suc}} \coloneqq  (1-p_{\mathrm{err}}m)P_{\mathrm{suc}} \ge \delta,
\end{align}
the algorithm will have the same overhead scaling as before. The expected number of calls to the oracle $ m $ needed as a function of $c(\ket{\psi})$, {will be, using \thmref{thm:coherence_complexity_amp_estimation}
\begin{align}
   m   \approx  \frac{1}{\sqrt{c(\ket{\psi})}}.
\end{align}

Therefore, we we expect our amplitude-amplified coherence detection protocol to be robust against noise so long as

\begin{align}
  p_{\mathrm{err}} \ll \sqrt{c(\ket{\psi})}.
\end{align}
Roughly speaking, the probability of an error occurring after each call to the oracle must be much smaller than the amount of coherence in the state.

\section{Outlook}
\label{sec:conclusions}

In this work, we have developed protocols for amplitude-amplified coherence detection and estimation for unknown pure states. We have shown how having access to the unitary transformation (and its inverse) that prepares a certain coherent state $\ket{\psi}$, as opposed to copies of the state, can lead to a quadratic reduction in the sample complexity needed to detect or estimate coherence.  Our proposed algorithms make use of amplitude amplification protocols to enhance the probability of successfully detecting or estimating the coherence of the state. 

A natural avenue for future work would be to explore the extent to which these techniques can be applied to the detection of other non-classical resources in the context of specific resource theories beyond coherence, such as entanglement theory~\cite{Horodecki2009Entanglement} or resource theories of magic~\cite{Veitch2014stabilizer}, or conversely for testing separability and staberlizerness. Another non-classical resource of recent interest is basis-independent or set coherence~\cite{Designolle2021BasisIndep,Taira2023ExperimentalCertification}, for which we expect our lower bounds to immediately hold. A second obvious extension of the current work would be to generalise the results of this paper to the detection of coherence in mixed states, or more generally in quantum channels.

As we approach the regime of fault-tolerant quantum computing, where fully error-corrected logical qubits are readily available, it is no too far-fetched to expect that protocols such as amplitude amplification and estimation will be at our disposal as a standard primitive within the fault-tolerant toolkit. If such a reality comes into fruition, we would like to see an experimental implementation of our protocols for amplitude-amplified coherence detection and estimation.


\section{Acknowledgements}

We would like to thank Adi Makmal for interesting discussions on a different problem which led us to this topic. We also acknowledge funding from grants PID2024-162155OB-I00, PID2024-162141OB-I00 funded by MICIU/AEI/10.13039/501100011033 and, as appropriate, by “ERDF A way of making Europe”, by “ERDF/EU”, by the “European Union” or by the “European Union NextGenerationEU/PRTR”, as well as funding from Ministry for Digital Transformation and of Civil Service of the Spanish Government through projects, QUANTUM ENIA project call - Quantum Spain project, and by the European Union through the Recovery, Transformation and Resilience Plan - NextGenerationEU within the framework of the Digital Spain 2026 Agenda. MS also ackowledges support from Ayuda Ramón y Cajal 2021 (RYC2021-032032-I, MICIU/AEI/10.13039/501100011033, ESF+) as well as Project FEDER C-EXP-256-UGR23 Consejería de Universidad, Investigación e Innovación y UE Programa FEDER Andalucía 2021-2027.

\appendix

\section{Proof of the lower bound in \thmref{thm:coherence_complexity_amp_estimation}}
\label{appx:lower_bound_proof_unitary}

Our proof leverages a recent result from two-element unitary discrimination given in Ref.~\cite{Weggemans2025lowerboundsunitary}, where we are given black-box access to a unitary $U \in\{U_1,U_2\}$ and we are tasked with deciding whether $U= U_1$ or $U=U_2$ up to some probability of error. Let us briefly summarize the context. 
The authors consider quantum query algorithms for property testing (also known as C-testers, where C refers to a complexity class) according to the following model: beginning with the initial state $\ket{\phi}$, the tester can make queries to $U \in \{U_1 ,U_2\}$ interleaved with input-independent unitaries $V_i$, and finally a two-outcome measurement is made which decides whether $U= U_1$ or $U=U_2$ with probability $\ge 2/3$. A $\mathrm{C}$-tester with  $\mathrm{C}=\mathrm{BQP}$, corresponds to the case where the initial state is of the form $\ket{\phi}= \ket{0}^{\otimes \mathrm{poly}(n)}$, which is precisely the setting we consider in of \thmref{thm:coherence_complexity_amp_estimation} as outlined at the start of \secref{sec:amp_amp_coherence_detection}.
In this context, Theorem~1 of Ref.~\cite{Weggemans2025lowerboundsunitary} provides a lower bound on any such $\mathrm{BQP}$-tester (in fact, the result presented therein holds for any $\mathrm{C}$-tester such that $\mathrm{C} \subseteq \mathrm{QMA}(2)/\mathrm{qpoly}$), in terms of the \textit{diamond-norm distance} between the two unitary channels:
\begin{align}
   \frac{1}{2} \norm{\U_1 - \U_2}_\diamond \coloneqq \frac{1}{2} \max_{\tr(\rho)=1, \rho \ge 0} \norm{\U_1 (\rho) - \U_2(\rho)}_{1},
\end{align}
where $\norm{X}_{1} \coloneqq \sqrt{X^\dagger X}$ is the trace norm.
We restate the lower bound as we need it from this theorem in the following proposition.
\begin{proposition}[Diamond-norm lower bound for unitary channel discrimination~\cite{Weggemans2025lowerboundsunitary}]\label{prop:Diamond_norm} Let $U_1 , U_2 \in U(d)$ such that $0 \not\in \mathrm{conv}( \mathrm{eig} (U_1^\dagger U_2) )$. 
Now let $U \in \{ U_1, U_2 \}$ be a unitary to which one has black-box access, including controlled operations, applications of the inverse and a combination of both. Suppose one has to decide whether
 (i) $U = U_1$ or (ii) $U =  U_2$ holds, promised that either one of them is the case and 
 $\frac{1}{2}\norm{\U_1 - \U_2}_\diamond \le \varepsilon$ where $\U_i (\cdot) \coloneqq U_i (\cdot) U_i^\dagger$ for each $i \in \{1,2\}$.
 Then any $\mathrm{BQP}$-tester which decides with success probability $\ge  2/3$ whether (i) or (ii) holds, needs to make at least $m = \Omega(1/\varepsilon)$ queries to $U$.
\end{proposition}

With this in hand, we are now in a position to prove the lower bound in \thmref{thm:coherence_complexity_amp_estimation}, which for convenience we restate in the following lemma.

\begin{lemma}  Given black-box access to $U \in \{U_\psi, U_\psi^\dagger\}$, detecting coherence in the unknown state $\ket{\psi} = U_\psi \ket{0} \in \H$, up to fixed error probability $\delta < \frac{1}{2}$, requires 
\begin{align}
    m  = \Omega\left(  \frac{1}{\sqrt{c(\ket{\psi})}}\log \delta^{-1} \right),
\end{align}
 calls to $U$ on average.

\end{lemma}

\begin{proof}[Proof of lower bound]
Detecting coherence in $\ket{\psi}$ given black box access to $U_\psi$ is equivalent to the unitary discrimination task of distinguishing $U_\psi$ from the set $\mathcal{X} \coloneqq \{ U \} $  of \textit{all} unitaries in $U(d)$ which satisfy $U \ket{0} = e^{i \theta}\ket{i}$ for some $\ket{i} \in \I$ and some angle $\theta \in [0,2\pi)$. Lower bounds on the easier task of discriminating $U_\psi$ from the single unitary $ U_{k_{\max}} \in \X$, where $U_{k_{\max}}$ is any choice of unitary which satisfies $U_{k_{\max}} \ket{0} = \ket{k_{\max}}$, straightforwardly follow form lower bounds on the original task, so we proceed by proving a lower bound on this two-element discrimination task.

 Our allowed coherence testers are $\mathrm{BQP}$-testers. Therefore, we can apply Proposition~$\ref{prop:Diamond_norm}$. To this end, the diamond distance (induced by the diamond norm) between $\U_\psi(\cdot) \coloneqq U_\psi(\cdot) U_\psi^\dagger$ and $\U_{k_{\max}}(\cdot) \coloneqq U_{k_{\max}}(\cdot) U_{k_{\max}}^\dagger$ is upper bounded by the operator (i.e. spectral) norm distance as [e.g.~see Proposition~1.6 of \cite{10353133}]
\begin{align} \label{subeq:diamond_vs_operator}
   &\frac{1}{2}\norm{\U_\psi - \U_{k_{\max}}}_{\diamond} \notag \le \norm{U_\psi - U_{k_{\max}}},
\end{align}
where $\norm{X}$ is the largest singular value of $X$. Without loss of generality, up to a globally irrelevant phase, we can always write 
\begin{align}
  \ket{\psi} =  \cos \theta \ket{k_{\max}} + \sin \theta \ket{{\perp}},
\end{align}
in terms of some state $\ket{\perp}$ in the orthogonal complement of $\ket{k_{\max}}$, with $\cos \theta \equiv \sqrt{p_{k_{\max}}}$. Let us also define
\begin{align}
    \ket{\psi_\perp} \coloneqq - \sin \theta \ket{k_{\max}} + \cos \theta \ket{{\perp}},
\end{align}
which is normalised and satisfies $\braket{\psi}{\psi_\perp} = 0$. Let us complete $\ket{\psi}$ and $\ket{\psi_\perp}$ to the orthonormal basis $\{ \ket{\psi}, \ket{\psi_\perp}, \ket{u_2}, \dots \ket{u_{d-1}} \}$ for $\H$. We can choose to define the action of $U_\psi$ on the computational basis as
\begin{align}
    U_\psi \ket{0} &= \ket{\psi}; \notag \\
     U_\psi \ket{1} &= \ket{\psi_\perp}; \notag \\
      U_\psi \ket{i} &= \ket{v_i}; \quad \forall i \in \{2,\dots,d-1\}.
\end{align}
Since 
\begin{align}
    S \coloneqq \mathrm{span}\{ \ket{\psi}, \ket{\psi_\perp} \} =\mathrm{span}\{ \ket{k_{\max}}, \ket{{\perp}} \},
\end{align}  
$\{ \ket{k_{\max}}, \ket{{\perp}}, \ket{u_2}, \dots \ket{u_{d-1}} \}$ also forms a valid orthonormal basis for $\H$ and we are free to define $U_{k_{\max}}$ via its action on the incoherent basis as
\begin{align}
       U_{k_{\max}} \ket{0} &= \ket{{k_{\max}}}; \notag \\
     U_{k_{\max}} \ket{1} &=  \ket{{\perp}}; \notag \\
      U_{k_{\max}} \ket{i} &= \ket{v_i}; \quad \forall i \in \{2,\dots,d-1\}.
\end{align}
$U_\psi$ and $U_{k_{\max}}$ act identically on the orthogonal complement $S^{\perp}$ to $S$, and differ only on the 2D subspace $S$, where in the $\{\ket{k_{\max}}, \ket{{\perp}}\}$ basis we have 
\begin{align}
    U_{\psi|S} = R(\theta) = \begin{pmatrix}
        \cos \theta & - \sin \theta \\ 
        \sin \theta & \cos \theta
    \end{pmatrix}; \ U_{k_{\max}| S} = \id.
\end{align}
It follows that their difference is block-diagonal 
\begin{align}
    U_{k_{\max}} - U_\psi = \begin{pmatrix}
       \id - R(\theta) & 0 \\ 
        0 & 0
    \end{pmatrix}.
\end{align}
Consequently,
\begin{align}
    \norm{ U_{k_{\max}} - U_\psi} = \norm{\id - R(\theta)}.
\end{align}
The eigenvalues of $\id - R(\theta)$ are $1 - e^{ \pm i \theta}$, hence 
\begin{align}
     \norm{ U_{k_{\max}} - U_\psi} = \max_{\pm} \abs{1 - e^{ \pm i \theta}} = \abs{2 \sin \theta/2} .
\end{align}
From the trigonometric identity $\sin^2 \theta/2 = \frac{1}{2}(1- \cos \theta)$ and the definition $\cos \theta = \sqrt{p_{k_{\max}}}$ we obtain
\begin{align} \label{subeq:operator_normbound}
    \norm{ U_{k_{\max}} - U_\psi} =  \sqrt{2} \sqrt{1- \sqrt{p_k}} \le \sqrt{2}\sqrt{1- p_k}.
\end{align}
Combining \eqref{subeq:diamond_vs_operator} and \eqref{subeq:operator_normbound} gives
\begin{align}
    \frac{1}{2}\norm{\U_\psi - \U_{k_{\max}}}_{\diamond} \le \sqrt{2} \sqrt{c(\ket{\psi})}
\end{align}

Now, let us assume that $\frac{1}{2}\norm{\U_\psi - \U_{k_{\max}}}_{\diamond} = \varepsilon$, such that
\begin{align}
    \varepsilon \le  \sqrt{2} \sqrt{c(\ket{\psi})}. \label{subeq:upper_bound_eps}
\end{align}
Then according to Proposition~\ref{prop:Diamond_norm}, any $\mathrm{BQP}$-tester that decides with success probability $\ge 2/3$ between $U_\psi$ from $U_{k_{\max}}$, needs to make at least 
\begin{align} \label{subeq:m_bounds_omega}
   m = \Omega\left( 1/\varepsilon \right) \ge \Omega\left( 1/\sqrt{c(\ket{\psi}} \right)  , 
\end{align}queries to the black-box unitary. Repeating $r=O(\log(1/\delta))$ times, boosts the success probability to $\ge 1-\delta$, and needs to make at least
\begin{align} 
   m = \Omega\left(\log \delta^{-1}/\sqrt{c(\psi)}\right) , 
\end{align}
queries, as claimed.  \end{proof}

\bibliography{apssamp}

@article{scully:science03,
	author = {M.O. Scully AND M.S. Zubairy AND G.S. Agarwal AND H. Walther},
	journal = {Science},
	pages = {862},
	title = {Extracting work from a single heat bath via vanishing quantum coherence},
	url = {https://www.science.org/doi/10.1126/science.1078955},
	volume = {299},
	year = {2003}}

@article{dorfman:pnas13,
	author = {K.E. Dorfmana AND D.V. Voroninea AND S. Mukamel AND M.O. Scully},
	journal = {Proc. Natl. Acad. Sci.},
	number = {8},
	pages = {2746},
	title = {Photosynthetic reaction center as a quantum heat engine},
	volume = {110},
	year = {2013}}

@article{svidzinsky:pra11,
	author = {A.A. Svidzinsky AND K.E. Dorfman AND M. O. Scully},
	journal = {Phys. Rev. A},
	pages = {053818},
	title = {Enhancing photovoltaic power by Fano-induced coherence},
	volume = {84},
	year = {2011}}

@article{manzano:po13,
	author = {D. Manzano},
	journal = {PLoS ONE},
	number = {2},
	pages = {e57041},
	title = {Quantum transport in quantum networks and photosynthetic complexes at the steady state},
	volume = {8},
	year = {2013},
	}

@article{proctor:nrp25,
	author = {T. Proctor AND K. Young AND A. D. Baczewski AND R. Blume-Kohout },
	journal = {Nat. Rev. Phys. },
	pages = {105},
	title = {Benchmarking quantum computers},
	volume = {7},
	year = {2025}}

@misc{tejero:preprint24,
      title={Squeezing light to get non-classical work in quantum engines}, 
      author={A. Tejero and D. Manzano and P. I. Hurtado},
      year={2024},
      eprint={2408.15085},
      archivePrefix={arXiv},
      primaryClass={quant-ph},
      url={https://arxiv.org/abs/2408.15085}, 
}

@incollection{Brassard2002OGAmpAmp,
  author       = {Gilles Brassard and Peter Høyer and Michele Mosca and Alain Tapp},
  title        = {Quantum Amplitude Amplification and Estimation},
  booktitle    = {Quantum Computation and Quantum Information: A Millennium Volume},
  editor       = {Samuel J. Lomonaco, Jr.},
  series       = {Contemporary Mathematics},
  volume       = {305},
  publisher    = {American Mathematical Society},
  year         = {2002},
  pages        = {53--74},
  doi          = {10.1090/conm/305/05215},
}

@article{Napoli2016Robustness,
  title = {Robustness of Coherence: An Operational and Observable Measure of Quantum Coherence},
  author = {Napoli, Carmine and Bromley, Thomas R. and Cianciaruso, Marco and Piani, Marco and Johnston, Nathaniel and Adesso, Gerardo},
  journal = {Phys. Rev. Lett.},
  volume = {116},
  issue = {15},
  pages = {150502},
  numpages = {6},
  year = {2016},
  month = {Apr},
  publisher = {American Physical Society},
  doi = {10.1103/PhysRevLett.116.150502},
  url = {https://link.aps.org/doi/10.1103/PhysRevLett.116.150502}
}

@article{Veitch2014stabilizer,
doi = {10.1088/1367-2630/16/1/013009},
url = {https://dx.doi.org/10.1088/1367-2630/16/1/013009},
year = {2014},
month = {jan},
publisher = {IOP Publishing},
volume = {16},
number = {1},
pages = {013009},
author = {Veitch, Victor and Hamed Mousavian, S A and Gottesman, Daniel and Emerson, Joseph},
title = {The resource theory of stabilizer quantum computation},
journal = {New Journal of Physics}
}

@article{Horodecki2009Entanglement,
  title = {Quantum entanglement},
  author = {Horodecki, Ryszard and Horodecki, Pawe\l{} and Horodecki, Micha\l{} and Horodecki, Karol},
  journal = {Rev. Mod. Phys.},
  volume = {81},
  issue = {2},
  pages = {865--942},
  numpages = {0},
  year = {2009},
  month = {Jun},
  publisher = {American Physical Society},
  doi = {10.1103/RevModPhys.81.865},
  url = {https://link.aps.org/doi/10.1103/RevModPhys.81.865}
}

@article{Ma2021CoherenceWitness,
  title = {Detecting and estimating coherence based on coherence witnesses},
  author = {Ma, Zhao and Zhang, Zhou and Dai, Yue and Dong, Yuli and Zhang, Chengjie},
  journal = {Phys. Rev. A},
  volume = {103},
  issue = {1},
  pages = {012409},
  numpages = {7},
  year = {2021},
  month = {Jan},
  publisher = {American Physical Society},
  doi = {10.1103/PhysRevA.103.012409},
  url = {https://link.aps.org/doi/10.1103/PhysRevA.103.012409}
}

@article{Baumgratz2014QuantifyingCoherence,
  title = {Quantifying Coherence},
  author = {Baumgratz, T. and Cramer, M. and Plenio, M. B.},
  journal = {Phys. Rev. Lett.},
  volume = {113},
  issue = {14},
  pages = {140401},
  numpages = {5},
  year = {2014},
  month = {Sep},
  publisher = {American Physical Society},
  doi = {10.1103/PhysRevLett.113.140401},
  url = {https://link.aps.org/doi/10.1103/PhysRevLett.113.140401}
}

@inbook{Aaronson2020ApproximateCounting,
  author    = {Scott Aaronson and Patrick Rall},
  title     = {Quantum Approximate Counting, Simplified},
  booktitle = {Proceedings of the 2020 Symposium on Simplicity in Algorithms (SOSA)},
  pages     = {24--32},
  year      = {2020},
  publisher = {Society for Industrial and Applied Mathematics},
  doi       = {10.1137/1.9781611976014.5},
  url       = {https://epubs.siam.org/doi/abs/10.1137/1.9781611976014.5},
  eprint    = {https://epubs.siam.org/doi/pdf/10.1137/1.9781611976014.5}
}

@article{Suzuki2020AmpEst,
  author    = {Yohichi Suzuki and Shumpei Uno and Rudy Raymond and Tomoki Tanaka and Tamiya Onodera and Naoki Yamamoto},
  title     = {Amplitude estimation without phase estimation},
  journal   = {Quantum Information Processing},
  volume    = {19},
  number    = {2},
  pages     = {75},
  year      = {2020},
  doi       = {10.1007/s11128-019-2565-2},
  url       = {https://doi.org/10.1007/s11128-019-2565-2},
  issn      = {1573-1332}
}

@article{Grover1996,
  author    = {Lov K. Grover},
  title     = {A Fast Quantum Mechanical Algorithm for Database Search},
  journal   = {Proceedings of the 28th Annual ACM Symposium on Theory of Computing (STOC)},
  pages     = {212--219},
  year      = {1996},
  doi       = {10.1145/237814.237866},
}

@article{Grinko2021IterativeAE,
  author    = {Dmitry Grinko and Julien Gacon and Christa Zoufal and Stefan Woerner},
  title     = {Iterative quantum amplitude estimation},
  journal   = {npj Quantum Information},
  volume    = {7},
  number    = {1},
  pages     = {52},
  year      = {2021},
  doi       = {10.1038/s41534-021-00379-1},
  url       = {https://doi.org/10.1038/s41534-021-00379-1},
  issn      = {2056-6387}
}

@book{NielsenChuang2010,
  author    = {Michael A. Nielsen and Isaac L. Chuang},
  title     = {Quantum Computation and Quantum Information: 10th Anniversary Edition},
  publisher = {Cambridge University Press},
  year      = {2010},
  isbn      = {978-1107002173},
}

@article{Lostaglio2015Thermodynamic,
  author    = {Matteo Lostaglio and David Jennings and Terry Rudolph},
  title     = {Description of quantum coherence in thermodynamic processes requires constraints beyond free energy},
  journal   = {Nature Communications},
  year      = {2015},
  volume    = {6},
  number    = {1},
  pages     = {6383},
  doi       = {10.1038/ncomms7383},
  url       = {https://doi.org/10.1038/ncomms7383},
  issn      = {2041-1723}
}

@misc{Farré2025WitnessKeyGrowing,
      title={Secure Hybrid Key Growing via Coherence Witnessing and Bipartite Encoding}, 
      author={Pol Julià Farré and Chris Aaron Schneider and Christian Deppe},
      year={2025},
      eprint={2508.06294},
      archivePrefix={arXiv},
      primaryClass={quant-ph},
      url={https://arxiv.org/abs/2508.06294}, 
}

@misc{Montanaro2018Property,
      title={A Survey of Quantum Property Testing}, 
      author={Ashley Montanaro and Ronald de Wolf},
      year={2018},
      eprint={1310.2035},
      archivePrefix={arXiv},
      primaryClass={quant-ph},
      url={https://arxiv.org/abs/1310.2035}, 
}

@article{Wagner2024coherence,
  doi = {10.22331/q-2024-02-05-1240},
  url = {https://doi.org/10.22331/q-2024-02-05-1240},
  title = {Coherence and contextuality in a {M}ach-{Z}ehnder interferometer},
  author = {Wagner, Rafael and Camillini, Anita and Galv{\~{a}}o, Ernesto F.},
  journal = {{Quantum}},
  issn = {2521-327X},
  publisher = {{Verein zur F{\"{o}}rderung des Open Access Publizierens in den Quantenwissenschaften}},
  volume = {8},
  pages = {1240},
  month = feb,
  year = {2024}
}

@article{Designolle2021BasisIndep,
  title = {Set Coherence: Basis-Independent Quantification of Quantum Coherence},
  author = {Designolle, S\'ebastien and Uola, Roope and Luoma, Kimmo and Brunner, Nicolas},
  journal = {Phys. Rev. Lett.},
  volume = {126},
  issue = {22},
  pages = {220404},
  numpages = {6},
  year = {2021},
  month = {Jun},
  publisher = {American Physical Society},
  doi = {10.1103/PhysRevLett.126.220404},
  url = {https://link.aps.org/doi/10.1103/PhysRevLett.126.220404}
}

@Article{Zhang2024Quantifying,
AUTHOR = {Zhang, Lin and Chen, Liang and He, Qiliang and Zhang, Yeqi},
TITLE = {Quantifying Quantum Coherence Using Machine Learning Methods},
JOURNAL = {Applied Sciences},
VOLUME = {14},
YEAR = {2024},
NUMBER = {16},
ARTICLE-NUMBER = {7312},
URL = {https://www.mdpi.com/2076-3417/14/16/7312},
ISSN = {2076-3417},
DOI = {10.3390/app14167312}
}

@Article{Huang2023ExperimentalRelEnt,
AUTHOR = {Huang, Xufeng and Yuan, Yuan and Niu, Yueping and Gong, Shangqing},
TITLE = {Experimental Direct Measurement of the Relative Entropy of Coherence},
JOURNAL = {Photonics},
VOLUME = {10},
YEAR = {2023},
NUMBER = {9},
ARTICLE-NUMBER = {1004},
URL = {https://www.mdpi.com/2304-6732/10/9/1004},
ISSN = {2304-6732},
DOI = {10.3390/photonics10091004}
}

@article{Taira2023ExperimentalCertification,
author = {Taira Giordani  and Rafael Wagner  and Chiara Esposito  and Anita Camillini  and Francesco Hoch  and Gonzalo Carvacho  and Ciro Pentangelo  and Francesco Ceccarelli  and Simone Piacentini  and Andrea Crespi  and Nicolò Spagnolo  and Roberto Osellame  and Ernesto F. Galvão  and Fabio Sciarrino },
title = {Experimental certification of contextuality, coherence, and dimension in a programmable universal photonic processor},
journal = {Science Advances},
volume = {9},
number = {44},
pages = {eadj4249},
year = {2023},
doi = {10.1126/sciadv.adj4249},
URL = {https://www.science.org/doi/abs/10.1126/sciadv.adj4249},
eprint = {https://www.science.org/doi/pdf/10.1126/sciadv.adj4249}}

@article{Acin2001UnitaryDistinguishability,
  title = {Statistical Distinguishability between Unitary Operations},
  author = {Ac\'{\i}n, A.},
  journal = {Phys. Rev. Lett.},
  volume = {87},
  issue = {17},
  pages = {177901},
  numpages = {4},
  year = {2001},
  month = {Oct},
  publisher = {American Physical Society},
  doi = {10.1103/PhysRevLett.87.177901},
  url = {https://link.aps.org/doi/10.1103/PhysRevLett.87.177901}
}

@article{Duan2007UnitaryDiscrim,
  title = {Entanglement is Not Necessary for Perfect Discrimination between Unitary Operations},
  author = {Duan, Runyao and Feng, Yuan and Ying, Mingsheng},
  journal = {Phys. Rev. Lett.},
  volume = {98},
  issue = {10},
  pages = {100503},
  numpages = {4},
  year = {2007},
  month = {Mar},
  publisher = {American Physical Society},
  doi = {10.1103/PhysRevLett.98.100503},
  url = {https://link.aps.org/doi/10.1103/PhysRevLett.98.100503}
}

@article{Abbot2024Complexity_Indefinite,
  title = {Quantum query complexity of Boolean functions under indefinite causal order},
  author = {Abbott, Alastair A. and Mhalla, Mehdi and Pocreau, Pierre},
  journal = {Phys. Rev. Res.},
  volume = {6},
  issue = {3},
  pages = {L032020},
  numpages = {6},
  year = {2024},
  month = {Jul},
  publisher = {American Physical Society},
  doi = {10.1103/PhysRevResearch.6.L032020},
  url = {https://link.aps.org/doi/10.1103/PhysRevResearch.6.L032020}
}

@misc{huang2025querycomplexitiesquantumchannel,
      title={Query complexities of quantum channel discrimination and estimation: A unified approach}, 
      author={Zixin Huang and Johannes Jakob Meyer and Theshani Nuradha and Mark M. Wilde},
      year={2025},
      eprint={2511.10832},
      archivePrefix={arXiv},
      primaryClass={quant-ph},
      url={https://arxiv.org/abs/2511.10832}, 
}

@article{Weggemans2025lowerboundsunitary,
  doi = {10.22331/q-2025-04-18-1717},
  url = {https://doi.org/10.22331/q-2025-04-18-1717},
  title = {Lower {B}ounds for {U}nitary {P}roperty {T}esting with {P}roofs and {A}dvice},
  author = {Weggemans, Jordi},
  journal = {{Quantum}},
  issn = {2521-327X},
  publisher = {{Verein zur F{\"{o}}rderung des Open Access Publizierens in den Quantenwissenschaften}},
  volume = {9},
  pages = {1717},
  month = apr,
  year = {2025}
}

@article{Streltsov2017CoherenceReview,
  title = {Colloquium: Quantum coherence as a resource},
  author = {Streltsov, Alexander and Adesso, Gerardo and Plenio, Martin B.},
  journal = {Rev. Mod. Phys.},
  volume = {89},
  issue = {4},
  pages = {041003},
  numpages = {34},
  year = {2017},
  month = {Oct},
  publisher = {American Physical Society},
  doi = {10.1103/RevModPhys.89.041003},
  url = {https://link.aps.org/doi/10.1103/RevModPhys.89.041003}
}

@article{Streltsov2015Geometric,
  title = {Measuring Quantum Coherence with Entanglement},
  author = {Streltsov, Alexander and Singh, Uttam and Dhar, Himadri Shekhar and Bera, Manabendra Nath and Adesso, Gerardo},
  journal = {Phys. Rev. Lett.},
  volume = {115},
  issue = {2},
  pages = {020403},
  numpages = {6},
  year = {2015},
  month = {Jul},
  publisher = {American Physical Society},
  doi = {10.1103/PhysRevLett.115.020403},
  url = {https://link.aps.org/doi/10.1103/PhysRevLett.115.020403}
}

@article{Cramer2010Tomography,
  author    = {Marcus Cramer and Martin B. Plenio and Steven T. Flammia and Rolando Somma and David Gross and Stephen D. Bartlett and Olivier Landon-Cardinal and David Poulin and Yi-Kai Liu},
  title     = {Efficient quantum state tomography},
  journal   = {Nature Communications},
  year      = {2010},
  volume    = {1},
  number    = {1},
  pages     = {149},
  doi       = {10.1038/ncomms1147},
  url       = {https://doi.org/10.1038/ncomms1147},
  issn      = {2041-1723}
}

@article{Lewenstein2000EntanglementWitness,
  title = {Optimization of entanglement witnesses},
  author = {Lewenstein, M. and Kraus, B. and Cirac, J. I. and Horodecki, P.},
  journal = {Phys. Rev. A},
  volume = {62},
  issue = {5},
  pages = {052310},
  numpages = {16},
  year = {2000},
  month = {Oct},
  publisher = {American Physical Society},
  doi = {10.1103/PhysRevA.62.052310},
  url = {https://link.aps.org/doi/10.1103/PhysRevA.62.052310}
}

@article{Boyer1998TightBounds,
  author    = {Michel Boyer and Gilles Brassard and Peter Høyer and Alain Tapp},
  title     = {Tight bounds on quantum searching},
  journal   = {Fortschritte der Physik},
  volume    = {46},
  number    = {4‐5},
  pages     = {493--505},
  year      = {1998},
  doi       = {10.1002/(SICI)1521-3978(199806)46:4/5<493::AID-PROP493>3.0.CO;2-P},
  eprint    = {quant-ph/9605034},
  archivePrefix = {arXiv},
  primaryClass   = {quant-ph}
}

@article{Kada2008IdentityTesting,
doi = {10.1088/1751-8113/41/39/395309},
url = {https://doi.org/10.1088/1751-8113/41/39/395309},
year = {2008},
month = {sep},
publisher = {},
volume = {41},
number = {39},
pages = {395309},
author = {Kada, Masaru and Nishimura, Harumichi and Yamakami, Tomoyuki},
title = {The efficiency of quantum identity testing of multiple states},
journal = {Journal of Physics A: Mathematical and Theoretical}
}

@article{Terhal2000Bell,
title = {Bell inequalities and the separability criterion},
journal = {Physics Letters A},
volume = {271},
number = {5},
pages = {319-326},
year = {2000},
issn = {0375-9601},
doi = {https://doi.org/10.1016/S0375-9601(00)00401-1},
url = {https://www.sciencedirect.com/science/article/pii/S0375960100004011},
author = {Barbara M. Terhal}
}

@article{Giovannetti2006Metrology,
  title = {Quantum Metrology},
  author = {Giovannetti, Vittorio and Lloyd, Seth and Maccone, Lorenzo},
  journal = {Phys. Rev. Lett.},
  volume = {96},
  issue = {1},
  pages = {010401},
  numpages = {4},
  year = {2006},
  month = {Jan},
  publisher = {American Physical Society},
  doi = {10.1103/PhysRevLett.96.010401},
  url = {https://link.aps.org/doi/10.1103/PhysRevLett.96.010401}
}

@article{Gisin2002Cryptography,
  title = {Quantum cryptography},
  author = {Gisin, Nicolas and Ribordy, Gr\'egoire and Tittel, Wolfgang and Zbinden, Hugo},
  journal = {Rev. Mod. Phys.},
  volume = {74},
  issue = {1},
  pages = {145--195},
  numpages = {0},
  year = {2002},
  month = {Mar},
  publisher = {American Physical Society},
  doi = {10.1103/RevModPhys.74.145},
  url = {https://link.aps.org/doi/10.1103/RevModPhys.74.145}
}

@article{Hillery2016CoherenceDJ,
  title = {Coherence as a resource in decision problems: The Deutsch-Jozsa algorithm and a variation},
  author = {Hillery, Mark},
  journal = {Phys. Rev. A},
  volume = {93},
  issue = {1},
  pages = {012111},
  numpages = {6},
  year = {2016},
  month = {Jan},
  publisher = {American Physical Society},
  doi = {10.1103/PhysRevA.93.012111},
  url = {https://link.aps.org/doi/10.1103/PhysRevA.93.012111}
}

@Article{Liu2019CoherenceAlgorithms,
AUTHOR = {Liu, Ye-Chao and Shang, Jiangwei and Zhang, Xiangdong},
TITLE = {Coherence Depletion in Quantum Algorithms},
JOURNAL = {Entropy},
VOLUME = {21},
YEAR = {2019},
NUMBER = {3},
ARTICLE-NUMBER = {260},
URL = {https://www.mdpi.com/1099-4300/21/3/260},
PubMedID = {33266975},
ISSN = {1099-4300},
DOI = {10.3390/e21030260}
}

@article{Ahnefeld2022CoherenceShor,
  title = {Coherence as a Resource for Shor's Algorithm},
  author = {Ahnefeld, Felix and Theurer, Thomas and Egloff, Dario and Matera, Juan Mauricio and Plenio, Martin B.},
  journal = {Phys. Rev. Lett.},
  volume = {129},
  issue = {12},
  pages = {120501},
  numpages = {7},
  year = {2022},
  month = {Sep},
  publisher = {American Physical Society},
  doi = {10.1103/PhysRevLett.129.120501},
  url = {https://link.aps.org/doi/10.1103/PhysRevLett.129.120501}
}

@misc{Ahnefeld2025CoherencePhaseEst,
      title={Coherence as a resource for phase estimation}, 
      author={Felix Ahnefeld and Thomas Theurer and Martin B. Plenio},
      year={2025},
      eprint={2505.18544},
      archivePrefix={arXiv},
      primaryClass={quant-ph},
      url={https://arxiv.org/abs/2505.18544}, 
}

@article{Lami2024Gravity,
  title = {Testing the Quantumness of Gravity without Entanglement},
  author = {Lami, Ludovico and Pedernales, Julen S. and Plenio, Martin B.},
  journal = {Phys. Rev. X},
  volume = {14},
  issue = {2},
  pages = {021022},
  numpages = {47},
  year = {2024},
  month = {May},
  publisher = {American Physical Society},
  doi = {10.1103/PhysRevX.14.021022},
  url = {https://link.aps.org/doi/10.1103/PhysRevX.14.021022}
}

@article{Marletto2017Gravity,
  title = {Gravitationally Induced Entanglement between Two Massive Particles is Sufficient Evidence of Quantum Effects in Gravity},
  author = {Marletto, C. and Vedral, V.},
  journal = {Phys. Rev. Lett.},
  volume = {119},
  issue = {24},
  pages = {240402},
  numpages = {5},
  year = {2017},
  month = {Dec},
  publisher = {American Physical Society},
  doi = {10.1103/PhysRevLett.119.240402},
  url = {https://link.aps.org/doi/10.1103/PhysRevLett.119.240402}
}

@misc{Ghoshal2020WitnessGravity,
      title={Coherence as witness for quantumness of gravity}, 
      author={Ahana Ghoshal and Arun Kumar Pati and Ujjwal Sen},
      year={2020},
      eprint={1909.07244},
      archivePrefix={arXiv},
      primaryClass={quant-ph},
      url={https://arxiv.org/abs/1909.07244}, 
}

@article{Matsumura2022CoherenceGravity,
  doi = {10.22331/q-2022-10-11-832},
  url = {https://doi.org/10.22331/q-2022-10-11-832},
  title = {Role of matter coherence in entanglement due to gravity},
  author = {Matsumura, Akira},
  journal = {{Quantum}},
  issn = {2521-327X},
  publisher = {{Verein zur F{\"{o}}rderung des Open Access Publizierens in den Quantenwissenschaften}},
  volume = {6},
  pages = {832},
  month = oct,
  year = {2022}
}

@article{Tanaka2021AestNoise,
  author    = {Tanaka, Tomoki and Suzuki, Yohichi and Uno, Shumpei and Raymond, Rudy and Onodera, Tamiya and Yamamoto, Naoki},
  title     = {Amplitude estimation via maximum likelihood on noisy quantum computer},
  journal   = {Quantum Information Processing},
  year      = {2021},
  volume    = {20},
  number    = {9},
  pages     = {293},
  doi       = {10.1007/s11128-021-03215-9},
  url       = {https://doi.org/10.1007/s11128-021-03215-9},
  issn      = {1573-1332}
}

@ARTICLE{Herbert2024AmpEstNoise,
  author={Herbert, Steven and Williams, Ifan and Guichard, Roland and Ng, Darren},
  journal={IEEE Transactions on Quantum Engineering}, 
  title={Noise-Aware Quantum Amplitude Estimation}, 
  year={2024},
  volume={5},
  number={},
  pages={1-23},
  doi={10.1109/TQE.2024.3476929}}

@misc{Brown2020ampestNoise,
      title={Quantum Amplitude Estimation in the Presence of Noise}, 
      author={Eric G. Brown and Oktay Goktas and W. K. Tham},
      year={2020},
      eprint={2006.14145},
      archivePrefix={arXiv},
      primaryClass={quant-ph},
      url={https://arxiv.org/abs/2006.14145}, 
}

@book{Helstrom1976quantum,
  author       = {Carl W. Helstrom},
  title        = {Quantum Detection and Estimation Theory},
  series       = {Mathematics in Science and Engineering, Vol. 123},
  publisher    = {Academic Press},
  address      = {New York},
  year         = {1976},
  isbn10       = {0123400503},
  isbn13       = {9780123400505},
  pages        = {309}
}

@INPROCEEDINGS{10353133,
  author={Haah, Jeongwan and Kothari, Robin and O’Donnell, Ryan and Tang, Ewin},
  booktitle={2023 IEEE 64th Annual Symposium on Foundations of Computer Science (FOCS)}, 
  title={Query-optimal estimation of unitary channels in diamond distance}, 
  year={2023},
  volume={},
  number={},
  pages={363-390},
  doi={10.1109/FOCS57990.2023.00028}}

@article{Diaz2018,
  title = {Using and Reusing Coherence to Realize Quantum Processes},
  author = {Díaz, María García and Fang, Kun and Wang, Xin and Rosati, Matteo and Skotiniotis, Michalis and Calsamiglia, John and Winter, Andreas},
  year = {2018},
  journal = {Quantum},
  volume = {2},
  pages = {100},
  publisher = {Verein zur Förderung des Open Access Publizierens in den Quantenwissenschaften},
  doi = {10.22331/q-2018-10-19-100},
}

@article{Diaz2020,
  title = {Accessible Coherence in Open Quantum System Dynamics},
  author = {Díaz, María García and Desef, Benjamin and Rosati, Matteo and Egloff, Dario and Calsamiglia, John and Smirne, Andrea and Skotiniotis, Michaelis and Huelga, Susana F.},
  year = {2020},
  journal = {Quantum},
  volume = {4},
  pages = {249},
  publisher = {Verein zur Förderung des Open Access Publizierens in den Quantenwissenschaften},
  doi = {10.22331/q-2020-04-02-249}
}

\end{document}